\documentclass[11pt]{article}
\usepackage[T2A]{fontenc}
\usepackage{ alphalph,etoolbox }
\usepackage{amsthm, amsmath, amssymb, amsbsy, amscd, amsfonts, latexsym, euscript,epsf, bm}
\newtheorem{thm}{Theorem}[section]
\usepackage{bbold, bm}
\usepackage{epic,eepic}
\usepackage{texdraw}
\usepackage[dvips]{color}
\usepackage{color}
\usepackage{ dsfont }
\usepackage{graphicx}
\usepackage{exscale,relsize}
\usepackage{authblk}
\usepackage{url}
\usepackage{enumerate}
\usepackage{graphicx}
\newtheorem{proposition}[thm]{Proposition}

\newtheorem{theorem}[thm]{Theorem}

\theoremstyle{definition}
\newtheorem{definition}[thm]{Definition}
\newtheorem{remark}[thm]{Remark}
\newtheorem{example}[thm]{Example}

\DeclareMathOperator{\id}{id}

\newcommand{\dmp}{k}
\newcommand{\zsp}{\mathbb{Z}_{>0}}
\newcommand{\Mfd}{\mathbf{M}}
\newcommand{\xc}{\mathbf{x}}
\newcommand{\cl}{\colon}
\newcommand{\qqquad}{\qquad\quad}
\newcommand{\cprime}{\/{\mathsurround=0pt$'$}}

\title{Local Yang–Baxter correspondences and set-theoretical solutions to the
Zamolodchikov tetrahedron equation}
\date{}

\author{S. Igonin\thanks{s-igonin@yandex.ru}~}
\author{S. Konstantinou-Rizos\thanks{skonstantin84@gmail.com}}
\affil{Centre of Integrable Systems, P.G. Demidov Yaroslavl State University, Yaroslavl, Russia}

\patchcmd{\subequations}{\alph{equation}}{\alphalph{\value{equation}}}{}{}

\begin{document}

\maketitle

\begin{abstract}
We study tetrahedron maps, which are set-theoretical solutions to the Zamolodchikov tetrahedron equation, 
and their matrix Lax representations defined by the local Yang--Baxter equation.

Sergeev [S.M. Sergeev 1998 Lett. Math. Phys. 45, 113--119] presented classification results on three-dimensional tetrahedron maps obtained
from the local Yang--Baxter equation for a certain class of matrix-functions 
in the situation when the equation possesses a unique solution which determines a tetrahedron map.
In this paper, using correspondences arising from the local Yang--Baxter equation 
for some simple $2\times 2$ matrix-functions,
we show that there are (non-unique) solutions to the local Yang--Baxter equation 
which define tetrahedron maps that do not belong to the Sergeev list; 
this paves the way for a new, wider classification of tetrahedron maps.
We present invariants for the derived tetrahedron maps and prove Liouville integrability for some of them.

Furthermore, using the approach of solving correspondences arising from 
the local Yang--Baxter equation, we obtain several new birational tetrahedron maps, 
including maps with matrix Lax representations on arbitrary groups,
a $9$-dimensional map associated with a Darboux transformation 
for the derivative nonlinear Schr\"odinger (NLS) equation, 
and a $9$-dimensional generalisation of the $3$-dimensional Hirota map.
\end{abstract}

\bigskip

\hspace{.2cm} \textbf{PACS numbers:} 02.30.Ik, 02.90.+p, 03.65.Fd.

\hspace{.2cm} \textbf{Mathematics Subject Classification 2020:} 16T25, 81R12.


\hspace{.2cm} \textbf{Keywords:} Zamolodchikov tetrahedron equation, local Yang--Baxter equation, tetrahedron maps 

\hspace{2.4cm} on groups, Lax representations, noncommutative tetrahedron maps, Darboux 

\hspace{2.4cm} transformations, derivative NLS equation.

\section{Introduction}\label{intro}

This paper is devoted to tetrahedron maps, which are set-theoretical solutions 
to the Zamolodchikov tetrahedron equation~\cite{Zamolodchikov,Zamolodchikov-2}.
The tetrahedron equation is a higher-dimensional analogue of the famous quantum Yang--Baxter equation
and has applications in many diverse branches of physics and mathematics, 
including statistical mechanics, quantum field theories, combinatorics, 
low-dimensional topology, and the theory of integrable systems
(see, e.g.,~\cite{Bazhanov-Sergeev,Bazhanov-Sergeev-2,Kapranov-Voevodsky,Kashaev-Sergeev,Kassotakis-Kouloukas,Kassotakis-Tetrahedron,Nijhoff,TalalUMN21} and references therein).
The Yang--Baxter and tetrahedron equations are members of the family 
of $n$-simplex equations~\cite{bazh82,DMH15,Hietarinta,Sharygin-Talalaev,Nijhoff,Nijhoff-2},
and they correspond to the cases of $2$-simplex and $3$-simplex, respectively.
Presently, the relations of tetrahedron maps with integrable systems 
and with algebraic structures (including groups and rings) are a very active area of research 
(see, e.g.,~\cite{bardakov22,Doliwa-Kashaev,IKKRP,igkr22,Kassotakis-Tetrahedron,Sharygin-Talalaev,TalalUMN21,Yoneyama21}).

There is a well-known method~\cite{Kashaev-Sergeev,Sergeev,Doliwa-Kashaev}
for constructing tetrahedron maps by means of the local Yang--Baxter equation, 
which is recalled in Section~\ref{slybe}.
Sergeev \cite{ Sergeev}
presented classification results on three-dimensional tetrahedron maps derived
by this method from the local Yang--Baxter equation for a certain class of matrix-functions.
These classification results in~\cite{Sergeev} 
are obtained in the situation when the local Yang--Baxter equation possesses 
a unique solution which determines a map satisfying the tetrahedron equation, 
so one derives a tetrahedron map.

In the present paper we study the more general situation when 
the local Yang--Baxter equation possesses a (possibly infinite) family 
of solutions and determines a correspondence rather than a map.
In this situation, our strategy is to add some extra equations to
the local Yang--Baxter equation 
so that the obtained system of equations defines a map.
This approach derives tetrahedron maps which do not belong to 
the Sergeev list in~\cite{ Sergeev}.
This smooths the path to a potential wider  classification of tetrahedron maps. 
Moreover, we construct new tetrahedron maps of derivative NLS type and of Hirota type.

The main results in this work are the following.
\begin{itemize}
\item We demonstrate that solutions to the local Yang--Baxter equation 
in the form of correspondences give rise to tetrahedron maps which do not belong to 
the Sergeev list~\cite{Sergeev} 
associated with the local Yang--Baxter equation.

In particular, using such correspondences arising from the local Yang--Baxter equation 
for some simple $2\times 2$ matrix-functions, 
in Section~\ref{3D-tetrahedron_maps} we construct tetrahedron maps of the form 
$(\mathbb{C}\setminus 0)^3\to(\mathbb{C}\setminus 0)^3$ with Lax representations.
For concreteness, we define these maps on the set $(\mathbb{C}\setminus 0)^3$,
however one can replace $\mathbb{C}\setminus 0$ by any commutative group.
Furthermore, invariants for these maps are presented.

These maps can be linearised by a change of variables, 
and the corresponding linear tetrahedron maps can be found 
in the work of Hietarinta~\cite{Hietarinta} in a very different context.
However, invariants and Lax representations were not known for them.

\item Using the described approach of correspondences, in Section~\ref{smgr}
we construct new birational tetrahedron maps of the form $\mathbf{G}^3\to\mathbf{G}^3$
with Lax representations, where $\mathbf{G}$ is an arbitrary group.
For an arbitrary group~$\mathbf{G}$, these maps cannot be linearised.

Considering $\mathbf{G}$ to be noncommutative, in Section \ref{smgr},
we derive tetrahedron maps in noncommutative variables
which we call \emph{noncommutative tetrahedron maps}.
Other examples of noncommutative tetrahedron maps 
(on rings rather than groups) can be found in \cite{Doliwa-Kashaev,igkr22, Sokor-2022}.

\item The approach of solving correspondences arising from 
the local Yang--Baxter equation for matrix-functions can be applied 
to matrix-functions determined by some Darboux matrices.

Specifically, using a Darboux matrix which determines a Darboux transformation  
for the derivative NLS equation, 
we construct a new $9$-dimensional birational tetrahedron map 
via the associated local Yang--Baxter correspondence.
Also, after simplifying this matrix by a change of variables, 
we extract from the local Yang--Baxter correspondence
a new $9$-dimensional generalisation of the well-known 
$3$-dimensional Hirota map~\cite{Doliwa-Kashaev,Sergeev}.
Invariants for the constructed maps are also presented.
\end{itemize}

In Section~\ref{slint} we show that some of the tetrahedron maps
constructed in Section~\ref{3D-tetrahedron_maps} are Liouville integrable.
Section~\ref{conclusions} contains concluding remarks and 
describes possible directions for further research.

\begin{remark}
It is known that the local Yang--Baxter equation 
can be viewed as a ``Lax equation'' for tetrahedron maps 
(see, e.g.,~\cite{DMH15} and references therein).
There is also a ``Lax equation'' for Yang-Baxter maps~\cite{Veselov2,Kouloukas2}.

The idea to derive such maps from a correspondence determined by a Lax equation 
was first used in~\cite{Sokor-Kouloukas} for the construction of a Yang--Baxter map of KdV type.
Then, more systematically, local Yang--Baxter correspondences were employed in~\cite{Sokor-2020} 
for the construction of tetrahedron maps. 
At the same time, 
a tetrahedron map of Hirota type was derived via such a correspondence in~\cite{Doliwa-Kashaev}.

\end{remark}

\section{Tetrahedron maps on sets and commutative groups}
\label{stmcg}

\subsection{Tetrahedron maps and the local Yang--Baxter equation}
\label{slybe}


Let $X$ be a set. A \emph{tetrahedron map} is a map 
\begin{equation}
\notag
T\cl X^3\rightarrow X^3,\qquad
T(x,y,z)=\big(f(x,y,z),g(x,y,z),h(x,y,z)\big),\qquad x,y,z\in X,
\end{equation}
satisfying the (Zamolodchikov) \emph{tetrahedron equation}
\begin{equation}\label{tetr-eq}
    T^{123}\circ T^{145} \circ T^{246}\circ T^{356}=T^{356}\circ T^{246}\circ T^{145}\circ T^{123}.
\end{equation}
Here, $T^{ijk}\cl X^6\rightarrow X^6$ for $i,j,k=1,\ldots,6$, $i<j<k$, is the map 
acting as $T$ on the $i$th, $j$th, $k$th factors 
of the Cartesian product $X^6$ and acting identically on the remaining factors.
For instance,
$$
T^{246}(x,y,z,r,s,t)=\big(x,f(y,r,t),z,g(y,r,t),s,h(y,r,t)\big),\qqquad x,y,z,r,s,t\in X.
$$


To describe the main ideas of this paper, consider a matrix-function
\begin{equation}\label{matrix-L}
   {\rm L}(x)= \begin{pmatrix} 
a(x) & b(x)\\ 
c(x) & d(x)
\end{pmatrix}
\end{equation}
depending on a variable $x\in{X}$. Using~\eqref{matrix-L}, one defines the matrix-functions
\begin{equation}\label{Lij-mat}
   {\rm L}_{12}(x)=\begin{pmatrix} 
 a(x) &  b(x) & 0\\ 
c(x) &  d(x) & 0\\
0 & 0 & 1
\end{pmatrix},\qquad
 {\rm L}_{13}(x)= \begin{pmatrix} 
 a(x) & 0 & b(x)\\ 
0 & 1 & 0\\
c(x) & 0 & d(x)
\end{pmatrix}, \qquad
 {\rm L}_{23}(x)=\begin{pmatrix} 
   1 & 0 & 0 \\
0 & a(x) & b(x)\\ 
0 & c(x) & d(x)
\end{pmatrix}.
\end{equation}
The equation 
\begin{equation}\label{Lax-Tetra}
    {\rm L}_{12}(u){\rm L}_{13}(v){\rm L}_{23}(w)=
		{\rm L}_{23}(z){\rm L}_{13}(y){\rm L}_{12}(x),\quad
		x,y,z,u,v,w\in{X},
\end{equation}
can be called the Maillet--Nijhoff equation~\cite{Nijhoff,Nijhoff-2} in Korepanov's form, 
but it is usually called the \emph{local Yang--Baxter equation} in the literature.

A well-known method~\cite{Doliwa-Kashaev, Kashaev-Sergeev,Sergeev}
for constructing tetrahedron maps by means of~\eqref{Lax-Tetra} is as follows.
Suppose that for any given $x,y,z\in{X}$ equation~\eqref{Lax-Tetra} determines 
$u,v,w\in{X}$ uniquely as functions $u=u(x,y,z)$, $v=v(x,y,z)$, $w=w(x,y,z)$.
Then, if some non-degeneracy conditions are satisfied, the corresponding map 
$T\cl {X}^3\to{X}^3$, $(x,y,z)\mapsto (u(x,y,z),v(x,y,z),w(x,y,z)),$ is a tetrahedron map. 

Sergeev~\cite{ Sergeev}
presented classification results on three-dimensional tetrahedron maps derived
by this method for a certain class of matrix-functions~\eqref{matrix-L}.
We would like to emphasise that these classification results in~\cite{Sergeev} 
are obtained in the situation when $u,v,w$ are determined from equation~\eqref{Lax-Tetra}
uniquely as functions $u=u(x,y,z)$, $v=v(x,y,z)$, $w=w(x,y,z)$.

In this paper, we study the more general situation when, for given $x,y,z\in{X}$, 
there may be many values of $u,v,w\in{X}$ satisfying~\eqref{Lax-Tetra}.
Then, equation~\eqref{Lax-Tetra} defines a correspondence between 
${X}^3$ and~${X}^3$ rather than a map ${X}^3\to{X}^3$.
For instance, see Example~\ref{exmxx} below. In such cases, our approach is to add some extra equations to~\eqref{Lax-Tetra} 
so that the obtained system of equations on $x,y,z,u,v,w\in{X}$ defines 
a map of the form ${X}^3\to{X}^3$, $\,(x,y,z)\mapsto(u,v,w)$.
As discussed in Sections~\ref{3D-tetrahedron_maps} and~\ref{conclusions}, 
this method gives tetrahedron maps which do not belong to 
the Kashaev--Korepanov--Sergeev list in~\cite{Sergeev}.

It is known that the local Yang--Baxter equation~\eqref{Lax-Tetra} 
can be viewed as a ``Lax equation'' or ``Lax system'' for the tetrahedron equation 
(see, e.g.,~\cite{DMH15} and references therein).
Using the above-mentioned approach, 
we construct tetrahedron maps $T\cl {X}^3\to{X}^3$ with the following property: 
for $x,y,z,u,v,w\in{X}$ the relation $(u,v,w)=T(x,y,z)$ implies 
(but is not necessarily equivalent to) equation~\eqref{Lax-Tetra}.
Then, one can say that \eqref{matrix-L},~\eqref{Lij-mat},~\eqref{Lax-Tetra}
constitute a \emph{Lax representation} for such a map~$T$, 
and the matrix-function~\eqref{matrix-L} is a \emph{Lax matrix} for~$T$.
Similar notions of Lax representations and Lax matrices 
for Yang-Baxter maps are well known~\cite{Veselov2,Kouloukas2}.

\subsection{Tetrahedron maps on commutative groups}
\label{3D-tetrahedron_maps}

\begin{example}
\label{exmxx}
Let ${X}=\mathbb{C}\setminus 0$.
We consider \eqref{matrix-L} to be 
${\rm L}(x)=\begin{pmatrix}x & 0\\ 0 & \frac{1}{x}\end{pmatrix}$, 
$\,x\in\mathbb{C}\setminus 0$.
Then, the corresponding matrix-functions~\eqref{Lij-mat} are
\begin{equation}
\label{Aijxx}
   {\rm L}_{12}(x)=\begin{pmatrix} 
 x &  0 & 0\\ 
0 &  \frac{1}{x} & 0\\
0 & 0 & 1
\end{pmatrix},\qquad
 {\rm L}_{13}(x)= \begin{pmatrix} 
 x & 0 & 0\\ 
0 & 1 & 0\\
0 & 0 & \frac{1}{x}
\end{pmatrix}, \qquad
 {\rm L}_{23}(x)=\begin{pmatrix} 
   1 & 0 & 0 \\
0 & x & 0\\ 
0 & 0 & \frac{1}{x}
\end{pmatrix}.
\end{equation}
The local Yang--Baxter equation~\eqref{Lax-Tetra} for~\eqref{Aijxx} 
is equivalent to the system
\begin{equation}\label{3d-corr}
    u=\frac{x y}{v},\qqquad w=\frac{y z}{v}.
\end{equation}
For given $x,y,z\in\mathbb{C}\setminus 0$ 
there are infinitely many values of $u,v,w\in\mathbb{C}\setminus 0$ satisfying~\eqref{3d-corr}.
Thus, system \eqref{3d-corr} gives a correspondence between 
$(\mathbb{C}\setminus 0)^3$ and~$(\mathbb{C}\setminus 0)^3$.

The correspondence~\eqref{3d-corr} does not define a map 
of the form $(\mathbb{C}\setminus 0)^3\to(\mathbb{C}\setminus 0)^3$, 
$(x,y,z)\mapsto(u,v,w)$. In order to obtain a map, we add to~\eqref{3d-corr}
the extra equation $v=x$.
Equations~\eqref{3d-corr} with $v=x$ are equivalent to
\begin{equation}
\notag
u=y,\qqquad v=x,\qqquad w=\frac{yz}{x}.
\end{equation}
Therefore, we obtain the map 
$T_1\cl (\mathbb{C}\setminus 0)^3\to(\mathbb{C}\setminus 0)^3$,
$\,T_1(x,y,z)=\left(y,x,\dfrac{yz}{x}\right)$.
According to Proposition~\ref{4_maps} below, this is a tetrahedron map,
and several more tetrahedron maps are derived in a similar fashion.
\end{example}

\begin{proposition}\label{4_maps}
The following maps 
$T_i\cl(\mathbb{C}\setminus 0)^3\to(\mathbb{C}\setminus 0)^3$, $\,i=1,2,3,4$,
\begin{align}
&{T_1}(x,y,z)= \left(y,x,\frac{yz}{x}\right),\label{4_maps-a}\\
&{T_2}(x,y,z)= \left(\frac{xy}{z},z,y\right),\label{4_maps-b}\\
&{T_3}(x,y,z)=\left(x^2,\frac{y}{x},xz\right),\label{4_maps-c}\\
&{T_4}(x,y,z)=\left(xz,\frac{y}{z},z^2\right),\label{4_maps-d}
\end{align}
satisfy the tetrahedron equation. 
Thus, they are three-dimensional tetrahedron maps.

The maps $T_1$, $T_2$, $T_3$ and $T_4$ share the functionally 
independent invariants $I_1=xy$ and $I_2=yz$.
Furthermore, the maps $T_1$ and $T_2$ are involutive, 
whereas $T_3$ and $T_4$ are noninvolutive.

Finally, the maps $T_1$, $T_2$, $T_3$ and $T_4$ share the Lax representation
\begin{equation}\label{Lax-1}
    {\rm L}_{12}(u){\rm L}_{13}(v){\rm L}_{23}(w)= {\rm L}_{23}(z){\rm L}_{13}(y){\rm L}_{12}(x),
\end{equation}
where ${\rm L}_{12}$, ${\rm L}_{13}$, ${\rm L}_{23}$ given by~\eqref{Aijxx} 
are obtained from ${\rm L}(x)=\begin{pmatrix}x & 0\\ 0 & \frac{1}{x}\end{pmatrix}$, 
$\,x\in\mathbb{C}\setminus 0$.
\end{proposition}
\begin{proof}
Similarly to Example~\ref{exmxx}, 
each of the maps $T_1$, $T_2$, $T_3$ and $T_4$ 
is derived by adding to system~\eqref{3d-corr} the extra equations $v=x$, $v=z$,  $v=\dfrac{y}{x}$, $v=\dfrac{y}{z}$, respectively.

A straightforward computation shows that, for each $i=1,2,3,4$, the map~$T_i$ 
satisfies the tetrahedron equation.

Now, equation~\eqref{Lax-1} for~\eqref{Aijxx} is equivalent to~\eqref{3d-corr}.
By definition of the maps $T_1$, $T_2$, $T_3$, $T_4$, 
for each $i=1,2,3,4$, the relation $(u,v,w)=T_i(x,y,z)$ implies~\eqref{3d-corr}.
Therefore, \eqref{Lax-1} with ${\rm L}(x)=\begin{pmatrix}x & 0\\ 0 & \frac{1}{x}\end{pmatrix}$
is a Lax representation for each of these maps.

We have $(T_1\circ T_1)(x,y,z)=(x,y,z)$, $\,(T_2\circ T_2)(x,y,z)=(x,y,z)$,
$\,(T_3\circ T_3)(x,y,z)=\left(x^4,\dfrac{y}{x^3},x^3z\right)$,
$\,(T_4\circ T_4)(x,y,z)=\left(xz^3,\dfrac{y}{z^3},z^4\right)$.
Therefore, the maps $T_1$ and $T_2$ are involutive, 
while $T_3$ and $T_4$ are noninvolutive.
Since $(I_1\circ T_i)(x,y,z)=xy$ and $(I_2\circ T_i)(x,y,z)=yz$, for each $i=1,2,3,4$, 
the functions $I_1=xy$ and $I_2=yz$ are invariants for~$T_i$.
\end{proof}

Now, we consider \eqref{matrix-L} to be 
${\rm L}(x)=\begin{pmatrix}0 & x\\ \frac{1}{x} & 0\end{pmatrix}$, $\,x\in\mathbb{C}\setminus 0$.
Then, the corresponding matrix-functions~\eqref{Lij-mat} are
\begin{equation}
\label{Lijxx}
   {\rm L}_{12}(x)=\begin{pmatrix} 
 0 &  x & 0\\ 
\frac{1}{x} &  0 & 0\\
0 & 0 & 1
\end{pmatrix},\qquad
 {\rm L}_{13}(x)= \begin{pmatrix} 
 0 & 0 & x\\ 
0 & 1 & 0\\
\frac{1}{x} & 0 & 0
\end{pmatrix}, \qquad
 {\rm L}_{23}(x)=\begin{pmatrix} 
   1 & 0 & 0 \\
0 & 0 & x\\ 
0 & \frac{1}{x} & 0
\end{pmatrix}.
\end{equation}
The local Yang--Baxter equation~\eqref{Lax-Tetra} for~\eqref{Lijxx} 
is equivalent to the system
\begin{equation}
\label{3d-corr2}
    v=x z,\qqquad w=\frac{y}{u},
\end{equation}
which gives a correspondence between 
$(\mathbb{C}\setminus 0)^3$ and~$(\mathbb{C}\setminus 0)^3$.
As shown in Proposition~\ref{4_maps-2} below, 
in order to derive tetrahedron maps from this correspondence, 
we add extra equations to system~\eqref{3d-corr2}.
\begin{proposition}
\label{4_maps-2}
The following maps
$\hat{T}_i\cl(\mathbb{C}\setminus 0)^3\to(\mathbb{C}\setminus 0)^3$, $\,i=1,2,3,4$,
\begin{align}
&{\hat{T}_1}(x,y,z)= \left(x,xz,\frac{y}{x}\right),\label{4_maps-2-a}\\
&{\hat{T}_2}(x,y,z)= \left(\frac{y}{z},xz,z\right),\label{4_maps-2-b}\\
&{\hat{T}_3}(x,y,z)=\left(yz,xz,\frac{1}{z}\right),\label{4_maps-2-c}\\
&{\hat{T}_4}(x,y,z)=\left(\frac{1}{x},xz,xy\right),\label{4_maps-2-d}
\end{align}
satisfy the tetrahedron equation. Thus, they are three-dimensional tetrahedron maps.

The map $\hat{T}_1$ admits the functionally independent invariants $I_1=x$ and $I_2=yz$, 
whereas the map~$\hat{T}_2$ possesses the invariants $I_3=z$ and $I_4=xy$. 
Additionally, $\hat{T}_3$ admits the invariant $I_5=x+yz$, the map~$\hat{T}_4$ 
admits the invariant $I_6=z+x y$, and both $\hat{T}_3$ and  $\hat{T}_4$ share the invariant $I_7=x y z$.
The functions $I_5$, $I_6$, $I_7$ are functionally independent.

Furthermore, the maps $\hat{T}_1$, $\hat{T}_2$, $\hat{T}_3$, $\hat{T}_4$ 
are involutive and share the Lax representation
\begin{equation}\label{Lax-2}
    {\rm L}_{12}(u){\rm L}_{13}(v){\rm L}_{23}(w)= {\rm L}_{23}(z){\rm L}_{13}(y){\rm L}_{12}(x),
\end{equation}
where ${\rm L}_{12}$, ${\rm L}_{13}$, ${\rm L}_{23}$ 
given by~\eqref{Lijxx} are obtained from
${\rm L}(x)=\begin{pmatrix}0 & x\\ \frac{1}{x} & 0\end{pmatrix}$.
\end{proposition}
\begin{proof}
Similarly to the proof of Propostition~\ref{4_maps}, 
each of the maps $\hat{T}_1$, $\hat{T}_2$, $\hat{T}_3$, $\hat{T}_4$ 
is derived by adding to system~\eqref{3d-corr2} one extra equation.
These extra equations are $u=x$, $u=\dfrac{y}{z}$,  $u=yz$, $u=\dfrac{1}{x}$, respectively.
The statements of Proposition~\ref{4_maps-2} 
are proved by straightforward computations very similar to those 
discussed in the proof of Propostition~\ref{4_maps}.
\end{proof}

In Propositions~\ref{4_maps},~\ref{4_maps-2}
the maps~\eqref{4_maps-a}--\eqref{4_maps-d}, \eqref{4_maps-2-a}--\eqref{4_maps-2-d}
are defined on the set $(\mathbb{C}\setminus 0)^3$.
Here, one can replace $\mathbb{C}\setminus 0$ 
by any commutative group~$A$, 
so that \eqref{4_maps-a}--\eqref{4_maps-d}, \eqref{4_maps-2-a}--\eqref{4_maps-2-d} 
can be regarded as tetrahedron maps of the form $A^3\to A^3$.

\begin{remark}
\label{rlin}
The maps~\eqref{4_maps-a}--\eqref{4_maps-d}, 
\eqref{4_maps-2-a}--\eqref{4_maps-2-d} can be linearised
by the following change of variables
$$
(x,y,z)\rightarrow\big(e^{\tilde{x}},e^{\tilde{y}},e^{\tilde{z}}\big),\qqquad
(u,v,w)\rightarrow\big(e^{\tilde{u}},e^{\tilde{v}},e^{\tilde{w}}\big).
$$
The corresponding linear tetrahedron maps can be found 
in the work of Hietarinta \cite{Hietarinta}	(see also~\cite{bardakov22}).
However, invariants and Lax representations were not known for these maps.
Note also that the works \cite{bardakov22, Hietarinta} do not consider the local Yang--Baxter equation.
\end{remark}

\subsection{Liouville integrability}
\label{slint} 

In Definition~\ref{dli} below 
we recall the standard notion of Liouville integrability for maps on manifolds
(see, e.g.,~\cite{fordy14,igkr22,Sokor-Sasha,vesel1991} and references therein).
\begin{definition}
\label{dli}
Let $\dmp\in\zsp$. 
Let $\Mfd$ be a $\dmp$-dimensional manifold 
with (local) coordinates $\xc_1,\dots,\xc_\dmp$.
A (smooth or analytic) map $F\cl \Mfd\to \Mfd$ is said to be \emph{Liouville integrable} 
(or \emph{completely integrable}) if 
one has the following objects on the manifold~$\Mfd$.
\begin{itemize}
	\item A Poisson bracket $\{\,,\,\}$ which is 
	invariant under the map~$F$ and is of constant rank~$2r$ 
	for some positive integer~$r\le\dmp/2$ (i.e. the $\dmp\times\dmp$ matrix with the entries 
	$\{\xc_i,\xc_j\}$ is of constant rank~$2r$).
	The invariance of the bracket means the following.
	For any functions $g$, $h$ on~$\Mfd$ one has $\{g,h\}\circ F=\{g\circ F,\,h\circ F\}$.
To prove that the bracket is invariant,
it is sufficient to check $\{g,h\}\circ F=\{g\circ F,\,h\circ F\}$ for $g=\xc_i$, $\,h=\xc_j$, $\,i,j=1,\dots,\dmp$.
	
\item If $2r<\dmp$, then one needs also $\dmp-2r$ functions $C_s$, $\,s=1,\dots,\dmp-2r$,
which are invariant under~$F$ (i.e. $C_s\circ F=C_s$) 
and are Casimir functions (i.e. $\{C_s,g\}=0$ for any function~$g$).
	\item One has $r$ functions $I_l$, $\,l=1,\dots,r$, which are invariant under~$F$
	and are in involution with respect to the Poisson bracket (i.e. $\{I_{l_1},I_{l_2}\}=0$ 
	for all $l_1,l_2=1,\dots,r$).
	\item The functions $C_1,\dots,C_{\dmp-2r},I_1,\dots,I_r$  must be functionally independent.
\end{itemize}	
\end{definition}

\begin{proposition}
\label{plit1t2}
The tetrahedron maps~\eqref{4_maps-a},~\eqref{4_maps-b} defined on 
the manifold $(\mathbb{C}\setminus 0)^3$ are Liouville integrable.
\end{proposition}
\begin{proof}
On the manifold $(\mathbb{C}\setminus 0)^3$ with coordinates $x,y,z$ 
we consider the Poisson bracket defined as follows
\begin{gather}
\notag
\{x,y\}=0,\qqquad \{x,z\}=\frac12x(x - y),\qqquad \{y,z\}=-\frac12x(x - y).
\end{gather}
It is easy to verify that this bracket is invariant under the map~\eqref{4_maps-a}
and is of rank~$2$.
The function $C_1=x+y$ is a Casimir and is invariant under~\eqref{4_maps-a}. 
The function $I_1=x^{-1}z$ is invariant as well. 
 Clearly, $C_1$ and $I_1$ are functionally independent.
Thus, all conditions of Definition~\ref{dli} are satisfied with 
$\dmp=3$ and $r=1$, which shows that~\eqref{4_maps-a} is Liouville integrable.

To show that the map~\eqref{4_maps-b} is Liouville integrable, 
we consider the Poisson bracket defined by
\begin{gather}
\notag
\{x,y\}=-\frac12(y - z) z,\qqquad \{x,z\}=\frac12(y - z) z,\qqquad
\{y,z\}=0.
\end{gather}
This bracket is invariant under~\eqref{4_maps-b} and is of rank~$2$.
The function $\tilde{C}_1=y+z$ is a Casimir and is invariant under~\eqref{4_maps-b}. 
The function $\tilde{I}_1=z^{-1}x$ is invariant as well.
Furthermore, $\tilde{C}_1$ and $\tilde{I}_1$ are functionally independent.
Here, all conditions of Definition~\ref{dli} are also satisfied with $\dmp=3$ and $r=1$.
\end{proof}

\section{Tetrahedron maps on arbitrary groups}
\label{smgr}
In this section, we construct noncommutative analogues of some of the tetrahedron maps 
from Propositions~\ref{4_maps} and~\ref{4_maps-2}. 
The construction is similar: we employ the matrices used in Section~\ref{3D-tetrahedron_maps} 
and replace $x\in\mathbb{C}\setminus 0$ by $\bm{x}\in\mathbf{G}$ and $\dfrac{1}{x}\in\mathbb{C}\setminus 0$ by $\bm{x}^{-1}\in\mathbf{G}$, where $\mathbf{G}$ is an arbitrary group, 
which may be noncommutative.

Let ${X}=\mathbf{G}$. 
Taking the matrix-functions ${\rm L}_{12}$, ${\rm L}_{13}$, ${\rm L}_{23}$
obtained from ${\rm L}(\bm{x})=\begin{pmatrix}\bm{x} & 0\\ 0 & \bm{x}^{-1}\end{pmatrix}$,
$\,\bm{x}\in\mathbf{G}$, and substituting them in~\eqref{Lax-Tetra}, 
we derive that \eqref{Lax-Tetra} is equivalent to the following system of equations
\begin{equation}\label{noncomm-corr}
\bm{u}\bm{v}=\bm{y}\bm{x},\qqquad 
\bm{u}^{-1}\bm{w}=\bm{z}\bm{x}^{-1},\qqquad \bm{v}^{-1}\bm{w}^{-1}=\bm{z}^{-1}\bm{y}^{-1}.
\end{equation}
Similarly to Example~\ref{exmxx}, we see that, for given 
$\bm{x},\bm{y},\bm{z}\in\mathbf{G}$, 
values of $\bm{u},\bm{v},\bm{w}\in\mathbf{G}$ satisfying~\eqref{noncomm-corr} 
are not determined uniquely. Thus, \eqref{noncomm-corr} gives a correspondence 
between $\mathbf{G}^3$ and~$\mathbf{G}^3$. 
The correspondence~\eqref{noncomm-corr} does not define a map 
of the form 
$\mathbf{G}^3\to\mathbf{G}^3$, $\,(\bm{x},\bm{y},\bm{z})\mapsto(\bm{u},\bm{v},\bm{w})$.
However, as shown in Theorem~\ref{NC_4_maps} below, 
one can obtain at least two tetrahedron maps of this form by adding suitable equations 
to system~\eqref{noncomm-corr}.

\begin{theorem}\label{NC_4_maps}
The following  maps $\mathbf{T}_i\cl\mathbf{G}^3\to\mathbf{G}^3$, $\,i=1,2$,
\begin{align}
&{\mathbf{T}_1}(\bm{x},\bm{y},\bm{z})= \left(\bm{y},\bm{x},\bm{y}\bm{z}\bm{x}^{-1}\right),\label{NC_maps-a}\\
&{\mathbf{T}_2}(\bm{x},\bm{y},\bm{z})= \left(\bm{y}\bm{x}\bm{z}^{-1},\bm{z},\bm{y}\right),\label{NC_maps-b}
\end{align}
satisfy the tetrahedron equation. 
Furthermore, these maps share the Lax representation
\begin{equation}\label{Lax}
    {\rm L}_{12}(\bm{u}){\rm L}_{13}(\bm{v}){\rm L}_{23}(\bm{w})= {\rm L}_{23}(\bm{z}){\rm L}_{13}(\bm{y}){\rm L}_{12}(\bm{x}),
\end{equation}
where ${\rm L}_{12}$, ${\rm L}_{13}$, ${\rm L}_{23}$ are obtained from
${\rm L}(\bm{x})=\begin{pmatrix}\bm{x} & 0\\ 0 & \bm{x}^{-1}\end{pmatrix}$, 
$\,\bm{x}\in\mathbf{G}$.

If the group $\mathbf{G}$ is noncommutative, 
the maps~\eqref{NC_maps-a},~\eqref{NC_maps-b} are noninvolutive.
\end{theorem}
\begin{proof}
In order to derive the map \eqref{NC_maps-a}, we add the equation $\bm{v}=\bm{x}$ to system \eqref{noncomm-corr}.
The map~\eqref{NC_maps-b} is obtained by adding to~\eqref{noncomm-corr}
the equation $\bm{v}=\bm{z}$.

This construction implies that the maps $\mathbf{T}_i$, $\,i=1,2$, 
share the same Lax representation~\eqref{Lax}, 
where ${\rm L}_{12}$, ${\rm L}_{13}$, ${\rm L}_{23}$ are obtained from
${\rm L}(\bm{x})=\begin{pmatrix}\bm{x} & 0\\ 0 & \bm{x}^{-1}\end{pmatrix}$. 
Substitution of~$\mathbf{T}_1$ to the lef-hand side and right-hand side 
of the tetrahedron equation~\eqref{tetr-eq} gives the following
\begin{align*}
(\mathbf{T}_1^{123}\circ \mathbf{T}_1^{145} & \circ \mathbf{T}_1^{246} \circ \mathbf{T}_1^{356})
(\bm{x}_1,\bm{x}_2,\bm{x}_3,\bm{x}_4,\bm{x}_5,\bm{x}_6)=\\
&=(\bm{x}_4,\bm{x}_2,\bm{x}_4\bm{x}_5\bm{x}_2^{-1},\bm{x}_1,\bm{x}_2\bm{x}_3\bm{x}_1^{-1},\bm{x}_4,\bm{x}_5\bm{x}_6\bm{x}_3^{-1}\bm{x}_2^{-1})=\\
&=(\mathbf{T}_1^{356}\circ \mathbf{T}_1^{246}\circ \mathbf{T}_1^{145}\circ \mathbf{T}_1^{123})
(\bm{x}_1,\bm{x}_2,\bm{x}_3,\bm{x}_4,\bm{x}_5,\bm{x}_6),\quad\text{for all}\quad
\bm{x}_i\in\mathbf{G},\quad i=1,\dots,6.
 \end{align*}
Hence, $\mathbf{T}_1$ is a tetrahedron map. 
Furthermore, 
$(\mathbf{T}_1\circ \mathbf{T}_1)(\bm{x},\bm{y},\bm{z})=
(\bm{x},\bm{y},\bm{x}\bm{y}\bm{z}\bm{x}^{-1}\bm{y}^{-1})$.
Therefore, if $\mathbf{G}$ is noncommutative, then 
$(\mathbf{T}_1\circ \mathbf{T}_1)(\bm{x},\bm{y},\bm{z})\neq (\bm{x},\bm{y},\bm{z})$,
which means that $\mathbf{T}_1$ is noninvolutive. 
The statements about $\mathbf{T}_2$ are proved similarly.
\end{proof}

Now, we consider \eqref{matrix-L} to be 
${\rm L}(\bm{x})=\begin{pmatrix}0 & \bm{x}\\ \bm{x}^{-1} & 0\end{pmatrix}$, $\,\bm{x}\in\mathbf{G}$.
Then, the corresponding matrix-functions~\eqref{Lij-mat} are
\begin{equation}
\label{Lijbmx}
   {\rm L}_{12}(\bm{x})=\begin{pmatrix} 
 0 &  \bm{x} & 0\\ 
\bm{x}^{-1} &  0 & 0\\
0 & 0 & 1
\end{pmatrix},\qqquad
 {\rm L}_{13}(\bm{x})= \begin{pmatrix} 
 0 & 0 & \bm{x}\\ 
0 & 1 & 0\\
\bm{x}^{-1} & 0 & 0
\end{pmatrix}, \qqquad
 {\rm L}_{23}(\bm{x})=\begin{pmatrix} 
   1 & 0 & 0 \\
0 & 0 & \bm{x}\\ 
0 & \bm{x}^{-1} & 0
\end{pmatrix}.
\end{equation}
The local Yang--Baxter equation~\eqref{Lax-Tetra} for~\eqref{Lijbmx}
is equivalent to the system
\begin{equation}\label{noncomm-corr-2}
\bm{u}\bm{w}=\bm{y},\qqquad  
\bm{u}^{-1}\bm{v}\bm{w}^{-1}=\bm{z}\bm{y}^{-1}\bm{x},
\qqquad \bm{v}^{-1}=\bm{z}^{-1}\bm{x}^{-1},
\end{equation}
which gives a correspondence between $\mathbf{G}^3$ and~$\mathbf{G}^3$.
In Theorem~\ref{NC_4_maps-2} we derive tetrahedron maps 
from this correspondence by adding extra equations to system~\eqref{noncomm-corr-2}.
\begin{theorem}\label{NC_4_maps-2}
The following  maps $\hat{\mathbf{T}}_i\cl\mathbf{G}^3\to\mathbf{G}^3$, $\,i=1,2$,
\begin{align}
&{\hat{\mathbf{T}}_1}(\bm{x},\bm{y},\bm{z})= \left(\bm{x},\bm{x}\bm{z},\bm{x}^{-1}\bm{y}\right),\label{NC_maps-c}\\
&{\hat{\mathbf{T}}_2}(\bm{x},\bm{y},\bm{z})= \left(\bm{y}\bm{z}^{-1},\bm{x}\bm{z},\bm{z}\right),\label{NC_maps-d}
\end{align}
satisfy the tetrahedron equation. These maps are involutive and share the Lax representation
\begin{equation}\label{Lax2}
    {\rm L}_{12}(\bm{u}){\rm L}_{13}(\bm{v}){\rm L}_{23}(\bm{w})= {\rm L}_{23}(\bm{z}){\rm L}_{13}(\bm{y}){\rm L}_{12}(\bm{x}),
\end{equation}
where ${\rm L}_{12}$, ${\rm L}_{13}$, ${\rm L}_{23}$  
are obtained from
${\rm L}(\bm{x})=\begin{pmatrix}0 & \bm{x}\\ \bm{x}^{-1} & 0\end{pmatrix}$, 
$\,\bm{x}\in\mathbf{G}$.
\end{theorem}
\begin{proof}
In order to derive the map \eqref{NC_maps-c}, we add to system~\eqref{noncomm-corr-2}
the equation $\bm{u}=\bm{x}$.
The map~\eqref{NC_maps-d} is obtained by adding to~\eqref{noncomm-corr-2}
the equation $\bm{w}=\bm{z}$.
The statements in this theorem are proved by straightforward computations 
similar to those in the proof of Theorem~\ref{NC_4_maps}.
\end{proof}

\begin{remark}
It is easily seen that, in the case of a noncommutative group~$\mathbf{G}$,
the maps in Theorems~\ref{NC_4_maps},~\ref{NC_4_maps-2} 
lose some of the invariants that their commutative analogues in 
Propositions~\ref{4_maps},~\ref{4_maps-2} possess. 
\end{remark}

\begin{remark}
If we swap the order of the factors in 
the third and first components of the maps~\eqref{NC_maps-c} 
and~\eqref{NC_maps-d}, respectively, then we obtain the maps 
$\tilde{\mathbf{T}}_i\cl\mathbf{G}^3\to\mathbf{G}^3$, $\,i=1,2$,
\begin{gather}
\label{tilT}
{\tilde{\mathbf{T}}_1}(\bm{x},\bm{y},\bm{z})= 
\left(\bm{x},\bm{x}\bm{z},\bm{y}\bm{x}^{-1}\right),\qqquad
{\tilde{\mathbf{T}}_2}(\bm{x},\bm{y},\bm{z})= 
\left(\bm{z}^{-1}\bm{y},\bm{x}\bm{z},\bm{z}\right).
\end{gather}
It turns out that \eqref{tilT} are also tetrahedron maps.
If $\mathbf{G}$ is noncommutative, they are noninvolutive.
However, we do not know any nontrivial Lax representations for the maps~\eqref{tilT}.
\end{remark}

\section{A derivative NLS Darboux transformation and the local Yang--Baxter  equation}
\label{sdtyb}

\subsection{A Darboux transformation for the derivative NLS equation}

Consider a Lax operator of the form $\mathcal{L}=D_x+U$, 
where $U=U(p,q;\lambda)$ is a $2\times 2$ matrix-function.
Here $p=p(x,t)$ and $q=q(x,t)$ are scalar functions (called \emph{potential functions}) 
satisfying a certain NLS-type partial differential equation, 
and $\lambda\in\mathbb{C}$ is a parameter called the spectral parameter. 

Following~\cite{Sokor-Sasha,SPS}, 
we say that, in this case, a \emph{Darboux transformation}
is determined by an invertible $2\times 2$ matrix-function~$M$ (called a \emph{Darboux matrix}) 
such that
\begin{gather}
\label{lpq}
M \mathcal{L} M^{-1}=M\big(D_x+U(p,q;\lambda)\big)M^{-1}=D_x-D_x(M)M^{-1}+MUM^{-1}=
D_x+U(\tilde{p},\tilde{q};\lambda),
\end{gather}
where functions $\tilde{p},\tilde{q}$ satisfy the same NLS-type equation. 
The matrix~$M$ may depend on the potential functions $p,q,\tilde{p},\tilde{q}$,
the parameter~$\lambda$, and some auxiliary functions.

Consider the Lax operator
\begin{gather}
\notag
\mathcal{L}=D_x+\lambda^{2}  
\begin{pmatrix}
      1&0\\
      0&-1
\end{pmatrix}
+\lambda\begin{pmatrix} 
0 & 2p \\ 
2q & 0
\end{pmatrix}, 
\end{gather}
which is the spatial part of a Lax pair for the derivative NLS equation~\cite{KN}
\begin{gather}
\label{ptqt}
p_t=p_{xx}+4\, (p^2q)_x,\qqquad 
q_t=-q_{xx}+4(pq^2)_x.
\end{gather}
The paper~\cite{SPS} presents a family of 
Darboux transformations corresponding to Lax operator.
This family contains a Darboux transformation with the following Darboux matrix
\begin{gather} \label{Z2-DM}
{\rm M}(f, p, \tilde{q}) = \lambda^2 
\begin{pmatrix} 
f & 0 \\ 0 & 0
\end{pmatrix} + 
\lambda \begin{pmatrix} 
0 & f p \\ f \tilde{q} & 0
\end{pmatrix} + 
\begin{pmatrix} 
1 & 0 \\ 0 & 0
\end{pmatrix}, 
\end{gather}
where $f$ is an auxiliary function 
which appears in the derivation of the Darboux matrix~${\rm M}$. 
The functions $p,q$ and $\tilde{p},\tilde{q}$ satisfy the derivative NLS equation~\eqref{ptqt}.
Also, $p,q,\tilde{p},\tilde{q},f$ obey the equations
$$
\partial_x p = 2 p \left( \tilde{p} \tilde{q}-p q\right) + \frac{2 p}{f},  \qquad
\partial_x \tilde{q} = 2 \tilde{q} \left(\tilde{p} \tilde{q}-p q\right) -\frac{2 \tilde{q}}{f},
 \qquad \partial_x f = 2 f (pq-\tilde{p} \tilde{q}),
$$
which constitute the spatial part of the B\"acklund transformation 
corresponding to this Darboux transformation.

\subsection{A tetrahedron map associated with the Darboux matrix}

In~\eqref{Z2-DM} we substitute $(f,p,\tilde{q},\lambda)\rightarrow (x_1,x_2,x_3,1)$,
$\,x_i\in\mathbb{C}$, $\, i=1,2,3$. Then the matrix~\eqref{Z2-DM} becomes
\begin{equation} \label{Lax_DNLS}
{\rm M}(x_1,x_2,x_3) = 
\begin{pmatrix} 
x_1+1 & x_1x_2  \\ x_1x_3 & 0
\end{pmatrix}.
\end{equation}

Using the construction of~\eqref{Lij-mat} from~\eqref{matrix-L}, 
from the matrix-function~\eqref{Lax_DNLS} we obtain the $3\times 3$ matrix-functions
\begin{gather*}
   {\rm M}_{12}(x_1,x_2,x_3)=\begin{pmatrix} 
 x_1+1 &  x_1x_2 & 0\\ 
x_1x_3 &  0 & 0\\
0 & 0 & 1
\end{pmatrix},\qqquad
 {\rm M}_{13}(x_1,x_2,x_3)=\begin{pmatrix} 
 x_1+1 & 0 & x_1x_2\\ 
0 & 1 & 0\\
x_1x_3 & 0 & 0
\end{pmatrix},\\ 
{\rm M}_{23}(x_1,x_2,x_3)=\begin{pmatrix} 
   1 & 0 & 0 \\
0 & x_1+1 & x_1x_2\\ 
0 & x_1x_3 & 0 
\end{pmatrix}.
\end{gather*}

The corresponding  local Yang--Baxter equation reads
$$
{\rm M}_{12}(u_1,u_2,u_3){\rm M}_{13}(v_1,v_2,v_3){\rm M}_{23}(w_1,w_2,w_3)
={\rm M}_{23}(z_1,z_2,z_3){\rm M}_{13}(y_1,y_2,y_3)
{\rm M}_{12}(x_1,x_2,x_3).
$$

The above is equivalent to the system
\begin{subequations}\label{DNLS-corr}
\begin{align}
    &u_2=\frac{x_1^2x_2x_3(1+y_1)(1+z_1)}{u_1(1+w_1)\left[x_1x_3(1+z_1)+(1+x_1)y_1y_3z_1z_2\right]},\label{DNLS-corr-a}\\
    &u_3=\frac{(1+u_1)\left[x_1x_3(1+z_1)+(1+x_1)y_1y_3z_1z_2\right]}{u_1(1+x_1)(1+y_1)},\label{DNLS-corr-b}\\
    &v_1=\frac{x_1-u_1+(1+x_1)y_1}{1+u_1},\label{DNLS-corr-c}\\
    &v_3=\frac{(1+u_1)x_1x_3z_1z_3}{x_1-u_1+y_1+x_1y_1},\label{DNLS-corr-d}\\
    &w_2=\frac{(1+w_1)y_1y_2\left[x_1x_3(1+z_1)+(1+x_1)y_1y_3z_1z_2\right]}{w_1x_1^2x_2x_3(1+y_1)(1+z_1)},\label{DNLS-corr-e}\\
    &w_3=\frac{x_1(1+x_1)x_2y_1(1+y_1)y_3z_1z_2}{v_2w_1(x_1-u_1+y_1+x_1y_1)\left[x_1x_3(1+z_1)+(1+x_1)y_1y_3z_1z_2\right]}.\label{DNLS-corr-f}
\end{align}
\end{subequations}

Similarly to the examples from the previous sections, 
system~\eqref{DNLS-corr} gives a correspondence between 
$\mathbb{C}^9$ and $\mathbb{C}^9$. 
In order to obtain a map of the form
$$
\big((x_1,x_2,x_3),(y_1,y_2,y_3),(z_1,z_2,z_3)\big)\mapsto
\big((u_1,u_2,u_3),(v_1,v_2,v_3),(w_1,w_2,w_3)\big),
$$
we need to add to~\eqref{DNLS-corr} three extra equations which would allow us 
to express $u_1$, $v_2$, $w_1$ in terms of $x_i$, $y_i$, $z_i$, $\,i=1,2,3$.
The choice of extra equations we will make below is motivated 
by the fact that we want to construct a tetrahedron map with enough (for integrability purposes) invariants. 

Indeed, note that \eqref{DNLS-corr-a} and \eqref{DNLS-corr-b} imply 
$$
u_2u_3=\frac{x_1^2(1+u_1)(1+z_1)x_2x_3}{u_1^2(1+x_1)(1+w_1)}.
$$
This suggests to add the equations $u_1=x_1$ and $w_1=z_1$, 
and then it follows that $u_2u_3=x_2x_3$, which provides us with three invariants: 
$x_1$, $z_1$, $x_2x_3$. 
Substituting $u_1=x_1$ and $w_1=z_1$ to \eqref{DNLS-corr-d}, \eqref{DNLS-corr-e}, 
\eqref{DNLS-corr-f}, one derives
\begin{gather}
\label{vvww}
v_2v_3=\frac{v_2x_1x_3z_1z_3}{y_1},\qqquad\quad w_2w_3=\frac{y_1y_2y_3z_2}{v_2x_1x_3z_1}.
\end{gather}
Now, if we add the equation $v_2=\dfrac{y_1y_2y_3}{x_1x_3z_1z_3}$, we obtain
$v_2v_3=y_2y_3$ and  $w_2w_3=z_2z_3$, which provides two more invariants.

To summarise, we have the following.

\begin{theorem}
\label{thtc9}
    The map $T$ given by
\begin{subequations}\label{DNLS-map}		
\begin{align}
T\big((x_1,x_2,x_3),(y_1,y_2,y_3),(z_1,z_2,z_3)\big)&=\big((u_1,u_2,u_3),(v_1,v_2,v_3),(w_1,w_2,w_3)\big),
\quad
x_i,y_i,z_i\in\mathbb{C},\\
    u_1&=x_1,\label{DNLS-map-a}\\
    u_2&=\frac{x_1x_2x_3(1+y_1)}{x_1x_3(1+z_1)+(1+x_1)y_1y_3z_1z_2},\label{DNLS-map-b}\\
    u_3&=\frac{x_1x_3(1+z_1)+(1+x_1)y_1y_3z_1z_2}{x_1(1+y_1)},\label{DNLS-map-c}\\
    v_1&=y_1,\label{DNLS-map-d}\\
    v_2&=\frac{y_1y_2y_3}{x_1x_3z_1z_3},\label{DNLS-map-e}\\
    v_3&=\frac{x_1x_3z_1z_3}{y_1},\label{DNLS-map-f}\\
    w_1&=z_1,\label{DNLS-map-g}\\
    w_2&=\frac{y_1y_2\left[x_1x_3(1+z_1)+(1+x_1)y_1y_3z_1z_2\right]}{x_1^2x_2x_3(1+y_1)z_1},\label{DNLS-map-h}\\
    w_3&=
		\frac{x_1^2x_2x_3(1+y_1)z_1z_2z_3}{y_1y_2\left[x_1x_3(1+z_1)+(1+x_1)y_1y_3z_1z_2\right]}\label{DNLS-map-i}
\end{align}
\end{subequations}
is  a nine-dimensional, birational, noninvolutive tetrahedron map. 
Moreover, it possesses  the following  functionally independent invariants
\begin{equation}
\label{iii}
I_1=x_1,\quad I_2=y_1,\quad I_3=z_1,\quad I_4=x_2x_3,\quad I_5=y_2y_3,\quad I_6=z_2z_3.
\end{equation}
\end{theorem}
\begin{proof}
The tetrahedron property for~$T$ 
can be readily verified with straightforward substitution to the tetrahedron equation. 
The map~$T$ is rational, invertible, and its inverse is rational as well.
Hence $T$ is birational.

We have
$$
(v_2\circ T)\big((x_1,x_2,x_3),(y_1,y_2,y_3),(z_1,z_2,z_3)\big)=
\frac{y_1^2y_2^2y_3}{x_1^2x_2x_3z_1^2z_2z_3}.
$$
Therefore, $T\circ T\neq\id$, which means that $T$ is noninvolutive.

Relations~\eqref{DNLS-map-a}, \eqref{DNLS-map-d}, \eqref{DNLS-map-g} 
say that the functions $I_1=x_1$, $\,I_2=y_1$, $\,I_3=z_1$ are invariant under the map~$T$.
The functions $I_4=x_2x_3$, $\,I_5=y_2y_3$, $\,I_6=z_2z_3$ are invariant as well, 
since we have $u_2u_3=x_2x_3$,  $v_2v_3=y_2y_3$, and $w_2w_3=z_2z_3$ in view of 
\eqref{DNLS-map-b}--\eqref{DNLS-map-c}, \eqref{DNLS-map-e}--\eqref{DNLS-map-f}, 
and \eqref{DNLS-map-h}--\eqref{DNLS-map-i}, respectively. 
Clearly, the invariants~\eqref{iii} are functionally independent.
\end{proof}

Several other tetrahedron maps associated with Darboux matrices
of NLS type equations were constructed in~\cite{Sokor-2020}

\subsection{Simplifications by a change of variables}
\label{ssimpl}
One can simplify the matrix-function~\eqref{Lax_DNLS}
by the following invertible change of variables
\begin{gather*}
\tilde{x}_1=x_1+1,\qqquad
\tilde{x}_2=x_1x_2,\qqquad
\tilde{x}_3=x_1x_3.
\end{gather*}
Then~\eqref{Lax_DNLS} becomes 
$\tilde{{\rm M}}(\tilde{x}_1,\tilde{x}_2,\tilde{x}_3)= 
\begin{pmatrix} 
\tilde{x}_1 & \tilde{x}_2 \\ \tilde{x}_3 & 0
\end{pmatrix}$.
To simplify notation, we write $x_1$, $x_2$, $x_3$ instead of 
 $\tilde{x}_1$, $\tilde{x}_2$, $\tilde{x}_3$ and obtain the matrix-function
\begin{equation} \label{tildLax}
\tilde{{\rm M}}(x_1,x_2,x_3)= 
\begin{pmatrix} 
x_1 & x_2 \\ x_3 & 0
\end{pmatrix}.
\end{equation}

Applying the described change of variables to each of the vectors
$$
(x_1,x_2,x_3),\quad(y_1,y_2,y_3),\quad(z_1,z_2,z_3),\quad
(u_1,u_2,u_3),\quad(v_1,v_2,v_3),\quad(w_1,w_2,w_3),
$$
from the statements of Theorem~\ref{thtc9} we obtain the following result.
\begin{theorem}
\label{thtiltc9}
The map $\tilde{T}$ given by
\begin{subequations}\label{tild-map}		
\begin{align}
\tilde{T}\big((x_1,x_2,x_3),(y_1,y_2,y_3),(z_1,z_2,z_3)\big)&=
\big((u_1,u_2,u_3),(v_1,v_2,v_3),(w_1,w_2,w_3)\big),\quad
x_i,y_i,z_i\in\mathbb{C},\\
    u_1&=x_1,\\
    u_2&=\frac{x_2 x_3 y_1}{x_1 y_3 z_2+x_3 z_1},\\
    u_3&=\frac{x_1 y_3 z_2+x_3 z_1}{y_1},\\
    v_1&=y_1,\\
    v_2&=\frac{y_2 y_3}{x_3 z_3},\\
    v_3&=x_3 z_3,\\
    w_1&=z_1,\\
    w_2&=\frac{y_2 (x_1 y_3 z_2+x_3 z_1)}{x_2 x_3 y_1},\\
    w_3&=\frac{x_2 x_3 y_1 z_2 z_3}{x_1 y_2 y_3 z_2+x_3 y_2 z_1}
\end{align}
\end{subequations}
is  a nine-dimensional, birational, noninvolutive tetrahedron map. 
Also, this map possesses the functionally independent invariants
\begin{equation}
\notag
I_1=x_1,\quad I_2=y_1,\quad I_3=z_1,\quad I_4=x_2x_3,\quad I_5=y_2y_3,\quad I_6=z_2z_3.
\end{equation}
\end{theorem}
The map~\eqref{DNLS-map} is obtained from the map~\eqref{tild-map} 
by the described invertible change of variables.
Thus \eqref{DNLS-map} and~\eqref{tild-map} are equivalent, 
but the formulas in~\eqref{tild-map} are much simpler than the ones in~\eqref{DNLS-map}.

\subsection{A generalised Hirota tetrahedron map}

Using the construction of~\eqref{Lij-mat} from~\eqref{matrix-L}, 
from the matrix-function~\eqref{tildLax} we obtain the $3\times 3$ matrix-functions
\begin{gather*}
   \tilde{{\rm M}}_{12}(x_1,x_2,x_3)\!=\!\begin{pmatrix} 
 x_1 &  x_2 & 0\\ 
x_3 &  0 & 0\\
0 & 0 & 1
\end{pmatrix},\,\ 
 \tilde{{\rm M}}_{13}(x_1,x_2,x_3)\!=\!\begin{pmatrix} 
 x_1 & 0 & x_2\\ 
0 & 1 & 0\\
x_3 & 0 & 0
\end{pmatrix},\,\ 
\tilde{{\rm M}}_{23}(x_1,x_2,x_3)\!=\!\begin{pmatrix} 
   1 & 0 & 0 \\
0 & x_1 & x_2\\ 
0 & x_3 & 0 
\end{pmatrix}.
\end{gather*}
The corresponding  local Yang--Baxter equation reads
$$
\tilde{{\rm M}}_{12}(u_1,u_2,u_3)\tilde{{\rm M}}_{13}(v_1,v_2,v_3)\tilde{{\rm M}}_{23}(w_1,w_2,w_3)
=\tilde{{\rm M}}_{23}(z_1,z_2,z_3)\tilde{{\rm M}}_{13}(y_1,y_2,y_3)
\tilde{{\rm M}}_{12}(x_1,x_2,x_3).
$$
The above is equivalent to the system
\begin{gather}
\label{ghc}
\begin{gathered}
u_1=\frac{u_3x_1y_1}{x_3z_1+x_1y_3z_2},\qquad 
v_1=\frac{x_3z_1+x_1y_3z_2}{u_3},\qquad v_2=\frac{x_2y_3z_2}{u_3w_3},\\
v_3=x_3z_3,\qquad  w_1=\frac{x_2x_3y_1z_1}{u_2(x_3z_1+x_1y_3z_2)},\qquad w_2=\frac{y_2}{u_2},
\end{gathered}
\end{gather}
which gives a correspondence between $\mathbb{C}^9$ and $\mathbb{C}^9$.
In order to obtain a map of the form
$$
\big((x_1,x_2,x_3),(y_1,y_2,y_3),(z_1,z_2,z_3)\big)\mapsto
\big((u_1,u_2,u_3),(v_1,v_2,v_3),(w_1,w_2,w_3)\big),
$$
we need to add to~\eqref{ghc} three extra equations which would allow us 
to express $u_2$, $u_3$, $w_3$ in terms of $x_i$, $y_i$, $z_i$, $\,i=1,2,3$.
We add the equations
\begin{gather}
\label{addeq}
u_2=x_2,\qqquad u_3=x_3,\qqquad w_3=z_3.
\end{gather}
The obtained system~\eqref{ghc},~\eqref{addeq}
determines a map described in Theorem~\ref{thghm} below.

\begin{theorem}[A generalised Hirota map]
\label{thghm}
The map $\mathbf{T}$ given by
\begin{subequations}\label{Hirota-gen}
\begin{align}
\mathbf{T}\big((x_1,x_2,x_3),(y_1,y_2,y_3),(z_1,z_2,z_3)\big)&
=\big((u_1,u_2,u_3),(v_1,v_2,v_3),(w_1,w_2,w_3)\big),
\quad
x_i,y_i,z_i\in\mathbb{C},\\
 u_1&=\frac{x_1x_3y_1}{x_3z_1+x_1y_3z_2},\label{Hirota-gen-a}\\
 u_2&=x_2,\label{Hirota-gen-b}\\
 u_3&=x_3,\label{Hirota-gen-c}\\
 v_1&=\frac{x_3z_1+x_1y_3z_2}{x_3},\label{Hirota-gen-d}\\
 v_2&=\frac{x_2y_3z_2}{x_3z_3},\label{Hirota-gen-e}\\
 v_3&=x_3z_3,\label{Hirota-gen-f}\\
 w_1&=\frac{x_3y_1z_1}{x_3z_1+x_1y_3z_2},\label{Hirota-gen-g}\\
 w_2&=\frac{y_2}{x_2},\label{Hirota-gen-h}\\
 w_3&=z_3,\label{Hirota-gen-i}
\end{align}
\end{subequations}
is  a nine-dimensional, birational, noninvolutive tetrahedron map. 
Moreover, it  possesses the following functionally independent invariants
\begin{equation}
\label{hinv}
I_1=x_2,\qquad I_2=x_3,\qquad I_3=z_3,\qquad I_4=x_1y_1,\qquad I_5=y_1z_1
\end{equation}
\end{theorem}
\begin{proof}
The tetrahedron property for~$\mathbf{T}$ 
can be readily verified with straightforward substitution to the tetrahedron equation. 
The map~$\mathbf{T}$ is rational, invertible, and its inverse is rational as well.
Hence $\mathbf{T}$ is birational.

We have
$$
(u_1\circ \mathbf{T})\big((x_1,x_2,x_3),(y_1,y_2,y_3),(z_1,z_2,z_3)\big)
=\frac{x_1x_2(x_3z_1+x_1y_3z_2)}{x_3(x_2z_1+x_1y_2z_3)}.
$$
Therefore, $\mathbf{T}\circ \mathbf{T}\neq\id$, which means that $\mathbf{T}$ is noninvolutive.

Relations~\eqref{Hirota-gen-b}, \eqref{Hirota-gen-c}, \eqref{Hirota-gen-i}
say that the functions $I_1=x_2$, $\,I_2=x_3$, $\,I_3=z_3$ 
are invariant under the map~$\mathbf{T}$.
The functions $I_4=x_1y_1$ and $I_5=y_1z_1$ are invariant as well, 
since we have $u_1v_1=x_1y_1$ and $v_1w_1=y_1z_1$ in view of 
\eqref{Hirota-gen-d}--\eqref{Hirota-gen-g} and \eqref{Hirota-gen-d}--\eqref{Hirota-gen-g}, respectively. 
Clearly, the invariants~\eqref{hinv} are functionally independent.
\end{proof}

\begin{remark}
\label{rinvol}
If in~\eqref{Hirota-gen} one substitutes $(x_2,x_3,y_2,y_3,z_2,z_3)\rightarrow (1,1,1,1,1,1)$
and denotes $x_1,y_1,z_1$ by $x,y,z$, then from~\eqref{Hirota-gen} one obtains the map 
\begin{equation}\label{Hirota}
    (x,y,z)\mapsto\left(\frac{xy}{x+z},x+z,\frac{yz}{x+z}\right),
\end{equation}
which is the well-known Hirota map~\cite{Doliwa-Kashaev,Sergeev}.
It is shown in~\cite{Doliwa-Kashaev} that the Hirota map~\eqref{Hirota} 
can be obtained from the local Yang--Baxter equation for the matrix-function
${\rm L}(x)=\begin{pmatrix} x & 1\\
 1 & 0 \end{pmatrix}$.
This matrix-function can be derived from~\eqref{tildLax} by the substitution 
$(x_2,x_3)\rightarrow (1,1)$.

It is worth noting that the Hirota map~\eqref{Hirota} is involutive,  
whereas its generalised version~\eqref{Hirota-gen} is not. 
Noninvolutive maps are more interesting than involutive ones,
since an involutive map~$\mathbb{T}$ has trivial dynamics, 
in the sense that the sequence 
$\mathbb{T},\,\mathbb{T}^2,\,\mathbb{T}^3,\dots,\mathbb{T}^n,\dots$ 
contains only~$\mathbb{T}$ and the identity map.
Thus, in this approach to dynamics,
the noninvolutive ``generalised'' Hirota map~\eqref{Hirota-gen} 
is more interesting than the classical Hirota map~\eqref{Hirota}.
\end{remark}

\section{Conclusions}
\label{conclusions}
In this paper, we showed that non-unique solutions to the local Yang--Baxter equation 
(solutions derived from correspondences)
give rise to tetrahedron maps which do not belong to the Sergeev classification.
This fact makes the ground for obtaining new classification results on tetrahedron maps. 

In particular, in this way, we derived the tetrahedron maps~\eqref{4_maps-a}--\eqref{4_maps-d}, 
\eqref{4_maps-2-a}--\eqref{4_maps-2-d} with Lax representations and invariants.
These maps do not belong to the Sergeev classification 
list in~\cite{ Sergeev}.
As discussed in Remark~\ref{rlin},
these maps can be linearised by a change of variables, 
and the corresponding linear tetrahedron maps can be found 
in the work of Hietarinta~\cite{Hietarinta} in a very different context, 
but invariants and Lax representations were not known for them. Also, we constructed new tetrahedron maps~\eqref{NC_maps-a}, \eqref{NC_maps-b}, 
\eqref{NC_maps-c}, \eqref{NC_maps-d} of the form $\mathbf{G}^3\to\mathbf{G}^3$
with Lax representations, where $\mathbf{G}$ is an arbitrary group.
Furthermore, we obtained new tetrahedron maps~\eqref{tilT}.
For an arbitrary group~$\mathbf{G}$, the birational 
maps~\eqref{NC_maps-a}, \eqref{NC_maps-b}, 
\eqref{NC_maps-c}, \eqref{NC_maps-d}, \eqref{tilT} cannot be linearised.

Several tetrahedron maps derived from correspondences 
satisfying the local Yang--Baxter equation were constructed in~\cite{Sokor-2020}. 
However, the matrices used in~\cite{Sokor-2020} were of the form~\eqref{matrix-L}, 
where $a,b,c,d$ were functions of a vector variable $\bar{x}=(x_1,x_2)$. 
Since some of the six-dimensional maps in~\cite{Sokor-2020} 
can be restricted on invariant leaves to three-dimensional maps,
which belong to the Sergeev list, 
the natural question arises as to whether all three-dimensional maps obtained 
by means of the $2\times 2$ matrix local Yang--Baxter equation 
are somehow equivalent to maps in the Sergeev list. 
In this paper, using correspondences arising from the local Yang--Baxter equation
for some simple $2\times 2$ matrices~\eqref{matrix-L} depending on a scalar variable~$x$, 
we derived three-dimensional maps~\eqref{4_maps-a}--\eqref{4_maps-d}, 
\eqref{4_maps-2-a}--\eqref{4_maps-2-d}, which are not 
equivalent to any of the maps of the Sergeev list.

Consider a set~$X$ and the permutation map
$P^{13}\cl X^3\to X^3$, 
$\,P^{13}(a_1,a_2,a_3)=(a_3,a_2,a_1)$, 
$\,a_i\in X$.
It is known that, if a map $T\cl X^3\to X^3$ satisfies the tetrahedron equation~\eqref{tetr-eq},
then the map $P^{13}\circ T\circ P^{13}$ obeys this equation as well.
A proof can be found, e.g., in~\cite{Kassotakis-Tetrahedron}.

If $T$ is one of the 
maps \eqref{4_maps-a}, \eqref{4_maps-c}, \eqref{4_maps-2-a}, 
\eqref{4_maps-2-c}, \eqref{NC_maps-a}, 
then $P^{13}\circ T\circ P^{13}$ is \eqref{4_maps-b}, \eqref{4_maps-d}, \eqref{4_maps-2-b}, \eqref{4_maps-2-d}, \eqref{NC_maps-b}, respectively.
Now, let $T$ be one of the maps \eqref{NC_maps-c}, \eqref{NC_maps-d}, \eqref{tilT}.
Computing $P^{13}\circ T\circ P^{13}$, we obtain the tetrahedron maps
\begin{gather}
\notag
(P^{13}{\hat{\mathbf{T}}_1}P^{13})(\bm{x},\bm{y},\bm{z})=
\left(\bm{z}^{-1}\bm{y},\bm{z}\bm{x},\bm{z}\right),\qqquad
(P^{13}{\hat{\mathbf{T}}_2}P^{13})(\bm{x},\bm{y},\bm{z})= 
\left(\bm{x},\bm{z}\bm{x},\bm{y}\bm{x}^{-1}\right),\\
\label{pttp}
(P^{13}{\tilde{\mathbf{T}}_1}P^{13})(\bm{x},\bm{y},\bm{z})= 
\left(\bm{y}\bm{z}^{-1},\bm{z}\bm{x},\bm{z}\right),\qqquad
(P^{13}{\tilde{\mathbf{T}}_2}P^{13})(\bm{x},\bm{y},\bm{z})= 
\left(\bm{x},\bm{z}\bm{x},\bm{x}^{-1}\bm{y}\right).
\end{gather}

Furthermore, it is known that, if a tetrahedron map~$T$ is invertible, 
then the inverse map~$T^{-1}$ satisfies the tetrahedron equation as well.
A proof can be found, e.g., in~\cite{igkr22}.

The noninvolutive tetrahedron maps~\eqref{NC_maps-a}, 
\eqref{NC_maps-b}, \eqref{tilT}, \eqref{pttp} are invertible, 
and the corresponding inverse maps
\begin{align*}
\mathbf{T}_1^{-1}(\bm{x},\bm{y},\bm{z})&= \left(\bm{y},\bm{x},\bm{x}^{-1}\bm{z}\bm{y}\right),\qquad
\mathbf{T}_2^{-1}(\bm{x},\bm{y},\bm{z})= \left(\bm{z}^{-1}\bm{x}\bm{y},\bm{z},\bm{y}\right),\\
\tilde{\mathbf{T}}_1^{-1}(\bm{x},\bm{y},\bm{z})&= 
\left(\bm{x},\bm{z}\bm{x},\bm{x}^{-1}\bm{y}\right),\qquad
\tilde{\mathbf{T}}_2^{-1}(\bm{x},\bm{y},\bm{z})=\left(\bm{y}\bm{z}^{-1},\bm{z}\bm{x},\bm{z}\right),\\
\label{ptp}
(P^{13}{\tilde{\mathbf{T}}_1}P^{13})^{-1}(\bm{x},\bm{y},\bm{z})&= 
\left(\bm{z}^{-1}\bm{y},\bm{x}\bm{z},\bm{z}\right),\qquad
(P^{13}{\tilde{\mathbf{T}}_2}P^{13})^{-1}(\bm{x},\bm{y},\bm{z})= 
\left(\bm{x},\bm{x}\bm{z},\bm{y}\bm{x}^{-1}\right).
\end{align*}
 are tetrahedron maps as well.

It is worth noting that the noncommutative versions~\eqref{NC_maps-a},~\eqref{NC_maps-b}
of the involutive maps~\eqref{4_maps-a},~\eqref{4_maps-b} are noninvolutive. 
As discussed in Remark~\ref{rinvol}, in terms of dynamics,
noninvolutive maps are more interesting than involutive ones.

\begin{remark}
For the maps constructed in this paper we proved the tetrahedron property by straightforward substitution to the tetrahedron equation. 
Here, straightforward computations are possible, since the maps are rather simple. 
For more complicated cases one can make use of the six-factorisation property 
which uses only the Lax representation of a map in order to prove its tetrahedron 
property~\cite{Sokor-2022}.
\end{remark}

The results of this paper can be extended in the following ways.
\begin{itemize}
    \item In Proposition~\ref{plit1t2} we showed that 
		the tetrahedron maps~\eqref{4_maps-a},~\eqref{4_maps-b} are Liouville integrable.
		One can try to prove Liouville integrability for other maps constructed in this paper,
		using the invariants presented here.
    \item The matrix local Yang--Baxter equation~\eqref{Lax-Tetra} can be regarded as a matrix refactorisation problem.
		The fact that one can derive tetrahedron maps from non-unique solutions 
		to correspondences determined by such matrix refactorisation problems
		is a quite important observation, which paths the way for new classification results on $n$-simplex maps, 
		since the usual approach to construction and classification of Yang--Baxter ($2$-simplex) and 
		tetrahedron ($3$-simplex) maps uses particular matrix refactorisation problems with unique solutions.
    \item  The method demonstrated in this paper can be employed for constructing set-theoretical  
		solutions to the $n$-simplex equation for any $n\ge 2$. Such results will appear in other publications.
		\item We suggest to study possible relations between the tetrahedron maps derived in this paper 
		and integrable systems of equations on the $3D$~lattice employing methods known from the literature 
		(for example, using symmetries~\cite{Kassotakis-Tetrahedron}).

  \item Since map \eqref{DNLS-map} admits integrals in separable variables, one may study the associated $3D$-lattice system, using the approach presented in \cite{Pavlos-Maciej-2}.
\end{itemize}

\section*{Acknowledgements}

The work on Sections~\ref{stmcg},~\ref{smgr},~\ref{sdtyb}
was supported by the Russian Science Foundation (grant No. 20-71-10110).
The work on Sections~\ref{intro},~\ref{conclusions}
was carried out within the framework of a development programme for the Regional Scientific and Educational Mathematical Centre of the P.G. Demidov Yaroslavl State University with financial support from the Ministry of Science and Higher Education of the Russian Federation (Agreement on provision of subsidy from the federal budget No. 075-02-2023-948).

We would like to thank V.~Bardakov and D.~Talalaev for useful discussions. 

Some computations were done using
NCAlgebra -- Non Commutative Algebra Package for Mathematica
\url{https://github.com/NCAlgebra/}  
\url{https://mathweb.ucsd.edu/~ncalg/}

\end{document}